\documentclass{llncs}
\usepackage[T1]{fontenc}
\usepackage{times}
\usepackage{subfigure}
\usepackage{lineno} 
\usepackage{paralist}
\usepackage{graphics,graphicx}

\usepackage{amssymb}
\usepackage{amsmath}
\usepackage{color}

\usepackage[textsize=footnotesize]{todonotes}
\graphicspath{{figures/}} % for figures in latex

\usepackage{todonotes}

\newtheorem{stat}{Claim}
\newtheorem{obs}{Observation}
\newtheorem{fact}{Fact}

\def\comment#1{}
\newcommand{\apex}{\mathop{\mathrm{apex}}}
\newcommand{\cone}{\mathop{\mathrm{cone}}}
\newcommand{\shell}{\mathop{\mathrm{shell}}}
\newcommand{\inter}{\mathop{\mathrm{in}}}
\newcommand{\out}{\mathop{\mathrm{out}}}

\begin{document}

\title{Drawing Planar Graphs with a Prescribed Inner Face}

\author{
  Tamara Mchedlidze,
  Martin N\"{o}llenburg,
	Ignaz Rutter
}

\institute{
  Karlsruhe Institute of Technology (KIT), Germany
}

\maketitle

\begin{abstract}
  Given a plane graph~$G$ (i.e., a planar graph with a fixed planar
  embedding) and a simple cycle~$C$ in~$G$ whose vertices are mapped
  to a convex polygon, we consider the question
  whether this drawing can be extended to a planar straight-line drawing
  of~$G$.  We characterize when this is possible in terms of simple
  necessary conditions, which we prove to be sufficient.  This also
  leads to a linear-time testing algorithm.  If a drawing extension
  exists, it can be computed in the same running time.
\end{abstract}

\section{Introduction}

The problem of extending a partial drawing of a graph to a complete
one is a fundamental problem in graph drawing that has many
applications, e.g., in dynamic graph drawing and graph interaction.  This problem has been
studied most in the planar setting and often occurs as a subproblem in the construction 
of planar drawings.

The earliest example of such a drawing extension result are so-called Tutte
embeddings.  In his seminal paper ``How to Draw a
Graph''~\cite{t-hdg-63}, Tutte showed that any triconnected planar
graph admits a convex drawing with its outer vertices fixed to an
arbitrary convex polygon.  The strong impact of this fundamental
result is illustrated by its, to date, more than 850 citations and the
fact that it received the ``Best Fundamental Paper Award'' from GD'12.
The work of Tutte has been extended in several ways.  In particular,
it has been strengthened to only require polynomial
area~\cite{dgk-pdhgg-11}, even in the presence of collinear
points~\cite{cegl-dgppo-12}.  Hong and Nagamochi extended the result
to show that triconnected graphs admit convex drawings when their
outer vertices are fixed to a star-shaped polygon~\cite{HongN08}.  For
general subdrawings, the decision problem of whether a planar straight-line
drawing extension exists is NP-hard~\cite{Patrignani06}.  Pach and
Wenger~\cite{PW98} showed that every subdrawing of a planar graph that
fixes only the vertex positions can be extended to a planar drawing
with~$O(n)$ bends per edge and that this bound is tight.  The drawing
extension problem has also been studied in a topological setting,
where edges are represented by non-crossing curves.  In contrast to
the straight-line variant, it can be tested in linear time whether a
drawing extension of a given subdrawing exists~\cite{ADFJKPR10}.
Moreover, there is a characterization via forbidden
substructures~\cite{jkr-ktppg-13}.

%\paragraph{\emph{\bf Contribution and Outline}} \todo{do we need the title in such a short intro?}
In this paper, we study the problem of finding a planar straight-line
drawing extension of a plane graph for which an arbitrary cycle has
been fixed to a convex polygon.  It is easy to see that
a drawing extension does not always exist in this case; see Fig.~\ref{fig:petal}. 
Let $G$ be a plane graph and let $C$ be a simple cycle of $G$ represented by
a convex polygon $\Gamma_C$ in the plane. The following two simple conditions are clearly necessary for the existence of a drawing extension: (i) $C$ has no chords that must be embedded outside of $C$ and  (ii) for every vertex $v$ with neighbors on $C$ that must be embedded outside of $C$ there exists a placement of $v$ outside $\Gamma_C$ such that the drawing of the graph induced by $C$ and $v$ is plane and bounded by the same cycle as in $G$. We show in this paper that these two conditions are in fact sufficient. Both conditions can be tested in linear time, and if they are satisfied, a corresponding drawing extension can be constructed within the same time bound.

Our paper starts with some necessary definitions (Section~\ref{sec:definitions}) and useful combinatorial properties (Section~\ref{sec:combinatorial}). The idea of our main result has two steps. We first show in Section~\ref{sec:extension-one-sided} that the conditions are sufficient if~$\Gamma_C$ is one-sided (i.e., it has an edge whose incident inner angles are both less than~$90^\circ$).  Afterward, we show in Section~\ref{sec:main-theorem} that, for an arbitrary convex polygon $\Gamma_C$, we can place the neighborhood of~$C$ in such a way that
the drawing is planar, and such that the boundary~$C'$ of its outer face is a
one-sided polygon~$\Gamma_{C'}$.  Moreover, our construction ensures
that the remaining graph satisfies the conditions for extendability
of~$\Gamma_{C'}$.  The general result then follows directly from the one-sided case. 

% 
% It is
% certainly necessary that we can place each neighbor of vertices of~$C$
% in the plane such that it sees all its neighbors on~$C$ without
% crossing~$\Gamma_C$.  We show that this condition is sufficient.  The
% proof works in two steps.  We first show that the condition is
% sufficient if~$\Gamma_C$ is one-sided (a precise definition is given
% in Section~\ref{sec:definitions}).  Afterwards, we show that, in the
% general case, we can place the neighborhood of~$C$ in such a way that
% the drawing is planar, and such that the outer cycle~$C'$ is a
% one-sided polygon~$\Gamma_{C'}$.  Moreover, our construction ensures
% that the remaining graph satisfies the conditions for extendability
% of~$\Gamma_{C'}$.  The result then follows from the one-sided case.
% 
% We give two simple necessary
% conditions for the existence of an extension and show that they are in fact
% sufficient.
% The conditions can be easily tested in linear time.  If
% they are satisfied, a corresponding drawing extension can be constructed
% within the same time bound.
% 
% The paper is structured as follows.  First, we present definitions and
% the precise necessary condition in Section~\ref{sec:definitions} and
% present structural results on the neighborhood of~$C$ in
% Section~\ref{sec:combinatorial}.  We treat the one-sided case in
% Section~\ref{sec:extension-one-sided} and finally prove the main
% result in Section~\ref{sec:main-theorem}.

\section{Definitions and a necessary condition}
\label{sec:definitions}

\paragraph{\emph{\bf Plane graphs and subgraphs}} A graph $G = (V,E)$ is \emph{planar} if it has a drawing $\Gamma$ in the plane $\mathbb R^2$
without edge crossings. Drawing $\Gamma$ subdivides the plane into
connected regions called \emph{faces}; the unbounded region is the
\emph{outer} and the other regions are the \emph{inner} faces.
The boundary of a face is called \emph{facial cycle}, and \emph{outer cycle} for the outer face. 
The cyclic ordering of edges around each vertex of $\Gamma$ together
with the description of the external face of $G$  is called
an \emph{embedding} of $G$. A graph $G$ with a planar
embedding is called \emph{plane graph}. A \emph{plane subgraph} $H$ of $G$ 
is a subgraph of $G$ together with a planar embedding that is the restriction of 
the embedding of $G$ to $H$. 

Let $G$ be a plane graph and let $C$ be a simple cycle of $G$. Cycle $C$ is called \emph{strictly internal}, if it does not contain any vertex of the outer face of $G$.
A chord of $C$ is called \emph{outer} if it lies outside $C$ in $G$. A cycle without outer chords is called \emph{outerchordless}. 
The \emph{subgraph of $G$ inside $C$} is the plane subgraph of $G$ that is constituted by vertices and edges of $C$ and all vertices and edges of $G$ that lie inside $C$.

%-----------------------------------------
\paragraph{\emph{\bf Connectivity}} A graph $G$ is \emph{$k$-connected}
if removal of any set of $k-1$ vertices of $G$ does not disconnect the graph. For $k=2,3$ a
$k$-connected graph is also called \emph{biconnected} and
\emph{triconnected}, respectively. 
An internally triangulated plane graph is triconnected if and only if there is no edge connecting two non-consecutive vertices of its outer cycle (see, for example,~\cite{Avis96}).

%-----------------------------------------
\paragraph{\emph{\bf Star-shaped and one-sided polygons}}
Let $\Pi$ be a polygon in the plane.  Two points inside or on the boundary of $\Pi$ are mutually \emph{visible}, if the straight-line segment
connecting them %does not intersect the boundary 
belongs to the interior of $\Pi$. The \emph{kernel} $K(\Pi)$ of polygon
$\Pi$ is the set of all the points inside $\Pi$ from which all vertices of
$\Pi$ are visible. We say that $\Pi$ is \emph{star-shaped} if $K(\Pi)\neq
\emptyset$. We observe that the given definition of a star-shaped polygon ensures that its kernel has a positive area.

A convex polygon $\Pi$ with $k$ vertices is called \emph{one-sided}, if there exists an edge $e$ (i.e., a line segment) of $\Pi$ such that the orthogonal projection to the line supporting $e$  maps all polygon vertices actually onto segment $e$. Then $e$ is called the \emph{base edge} of $\Pi$. Without loss of generality let $e=(v_1,v_k)$ and  $v_1, \dots, v_k$ be the clockwise ordered sequence of vertices of $\Pi$.

% an ordering of its vertices $v_1,\dots,v_k$ in the order they appear in $\Pi$, such that the orthogonal projections of $v_2,\dots,v_{k-1}$ on the segment 
% $\overline{v_1,v_k}$ appear in this particular order between $v_1$ and $v_k$. Edge $(v_1,v_k)$ is then referred to as \emph{base edge} of $\Pi$. 

%-----------------------------------------
\paragraph{\emph{\bf Extension of a drawing}}
Let $G$ be a plane graph and let $H$ be a plane subgraph of 
$G$.  Let $\Gamma_H$ be a planar straight-line drawing of $H$. We say that $\Gamma_H$ is 
\emph{extendable} if drawing $\Gamma_H$ can be 
completed to a planar straight-line drawing $\Gamma_G$ of the plane graph $G$. 
Then $\Gamma_G$ is called an \emph{extension} of $\Gamma_H$. 
A planar straight-line drawing of $G$ is called \emph{convex}, if every face of $G$ (including the outer face) is represented as a convex polygon.

The following theorem by Hong and Nagamochi~\cite{HongN08}
shows the extendability of a prescribed star-shaped outer face of a plane graph.
\begin{theorem}[Hong, Nagamochi~\cite{HongN08}]
\label{theorem:HongNagamochi} Every drawing  of the outer face $f$
of a $3$-connected graph $G$ as a star-shaped polygon can be
extended to a planar drawing of $G$, where each internal face is
represented by a convex polygon. Such a drawing can be computed in
linear time.
\end{theorem}
\begin{figure}[tb]
 \centering
 \subfigure[\label{fig:petal}]
 {\includegraphics[scale=0.7]{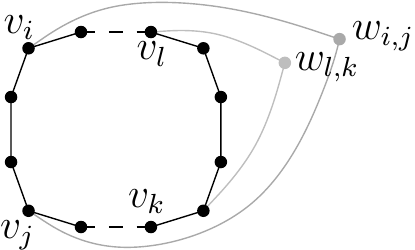}}
 \hspace{+1cm}
 \subfigure[\label{fig:geom_stat}]
 {\includegraphics[scale=0.7]{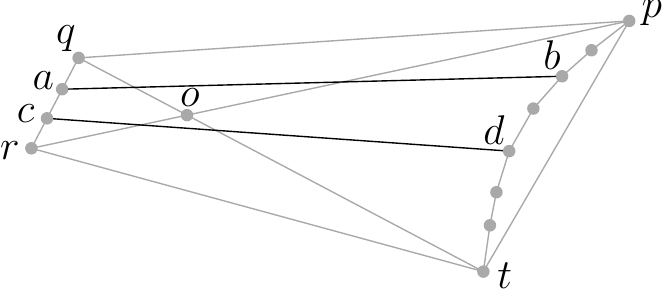}}
 \caption{Convex polygon of cycle $C$ is denoted by black. Vertex $w_{i,j}$ cannot be placed on the plane without changing the embedding or intersecting $C$.  Vertices $w_{i,j}$ and $w_{l,k}$ are petals of $C$, where $w_{l,k} \prec w_{i,j}$. Petal $w_{l,k}$ is realizable, while petal $w_{i,j}$ is not. (b) Illustration of Fact~\ref{stat:no_crossings}.}
\end{figure}

%-----------------------------------------
\vspace{-0.3cm}
\paragraph{\emph{\bf Petals and stamens}}
Let $G$ be a plane graph, and let $P_{uv}$ be 
a path in $G$ between vertices $u$ and $v$. Its subpath from vertex $a$ to vertex $b$ is denoted by $P_{uv}[a,b]$. 
Let $C$ be a simple cycle of $G$, and let $v_1,\dots,v_k$ be the vertices of $C$ in clockwise order. Given two vertices $v_i$ and $v_j$ of $C$, 
we denote by $C[v_i,v_j]$ the subpath of $C$ encountered when we traverse $C$ clockwise from $v_i$ to $v_j$. Assume that $C$ is represented by a convex polygon $\Gamma_C$ in the plane.  We say that a vertex $v_i$, $1\leq i \leq k$ of $\Gamma_C$ is \emph{flat}, if $\angle v_{i-1} v_i v_{i+1} = \pi$. Throughout this paper, we assume that convex polygons do not have flat vertices.  

A vertex $w\in V(G)\setminus V(C)$ adjacent to at least two vertices of $C$ and lying outside $C$ in $G$, is called a \emph{petal} of $C$ (see Figure~\ref{fig:petal}).  Consider the plane subgraph $G'$ of $G$ induced by the vertices $V(C) \cup \{w\}$. Vertex $w$ appears on the boundary of $G'$ between two vertices of $C$, i.e. after some $v_i \in V(C)$ and before some $v_j \in V(C)$ in clockwise order. To indicate this fact, we will denote petal $w$ by $w_{i,j}$. Edges $(w_{i,j},v_i)$ and $(w_{i,j},v_j)$ are called the \emph{outer} edges of petal $w_{i,j}$. The subpath $C[v_i,v_j]$ of $C$ is called \emph{base} of the petal $w_{i,j}$. %A vertex $v_f$, $i+1\leq f \leq j-1$ is called an \emph{internal} vertex of petal $w_{i,j}$. 
A vertex $v_f$ is called \emph{internal}, if it appears on $C$ after $v_i$ and before $v_j$ in clockwise order.
A petal $w_{i,i+1}$ is called \emph{trivial}. 
A vertex of $V(G)\setminus V(C)$ adjacent to exactly one vertex of $C$ is called a \emph{stamen} of $C$. 

Let $v$ be a petal of $C$ and let $u$ be either a petal or a stamen of $C$, we say that $u$ is \emph{nested} in $v$, and denote this fact by $u \prec v$, if 
$u$ lies in the cycle delimited by the base and the outer edges of petal $v$.  
For two stamens $u$ and $v$, neither $u \prec v$ nor $v\prec u$.
So for each pair of stamens or petals $u$ and $v$ we have either $u \prec v$, or $v \prec u$, or none of these.
This relation $\prec$ is a partial order. A petal or a stamen $u$ of $C$ is called \emph{outermost} if it is maximal with respect to $\prec$.  

%-----------------------------------------
\paragraph{\emph{\bf Necessary petal condition}} Let again $G$ be a plane graph and let $C$ be an outerchordless cycle of $G$ represented by a convex polygon $\Gamma_C$ in the plane. Let $w_{i,j}$ be a petal of $C$. Let $G'$ be the plane subgraph of $G$, induced by the vertices $V(C)\cup \{w_{i,j}\}$.  We
say that $w_{i,j}$ is \emph{realizable} if there exists a planar drawing 
of $G'$ which is an extension of $\Gamma_C$. 
This gives us the necessary condition that \emph{$\Gamma_C$ is extendable only if each petal of $C$ is realizable.} In the rest of the paper we prove that this condition is sufficient.

\section{Combinatorial Properties of Graphs and Petals}
\label{sec:combinatorial}

In this section, we derive several properties of petals in graphs,
which we  use throughout the construction of the drawing extension
in the remaining parts of this paper. 
Proposition~\ref{prop:triangulation} allows us to restrict our
attention to maximal plane graphs for which the given cycle~$C$ is
strictly internal. The remaining lemmas are concerned with the
structure of the (outermost) petals of~$C$ in such a graph.

\begin{proposition}
  \label{prop:triangulation}
  Let $G$ be a plane graph on $n$ vertices and let $C$ be a simple outerchordless cycle
  of $G$.  There exists a plane supergraph~$G'$ of~$G$
  with~$O(n)$ vertices such that
  \begin{inparaenum}[(i)]
  \item $G'$ is maximal plane,
  \item there are no outer chords of~$C$ in~$G'$,
  \item each petal of~$G'$ with respect to~$C$ is either trivial or
    has the same neighbors on~$C$ as in~$G$, and
  \item $C$ is strictly internal to~$G$.
  \end{inparaenum}
\end{proposition}
\begin{proof}
    First, we arbitrarily triangulate the graph in the interior
    of~$C$.  Clearly, this creates neither outer chords nor new
    petals.  For triangulating the graph outside of~$C$, we have to be
    more careful to avoid the creation of outer chords and potentially
    unrealizable petals.  We proceed in three steps.  First, we ensure
    that, for each edge~$e$ of~$C$, the incident face outside of~$C$
    is a triangle.  In this phase, we create new petals, all of which
    are realizable.  Second, we ensure that all faces incident to
    vertices of~$C$ are triangles without introducing any new petals.
    After this step, all faces incident to vertices of~$C$ are
    triangles.  We then triangulate the remaining faces arbitrarily.
    Afterward, it may still happen that~$C$ is not a strictly
    internal in the resulting graph.  However, in this case there is
    exactly one vertex of~$C$ incident to the outer face.  Then~$C$ is
    made strictly internal by adding a final new vertex in the outer
    face.  In the following, we describe these steps in more detail
    and argue their correctness.

    In the first step, we create, for each edge~$e$ of~$C$, a new
    vertex~$v_e$ that is adjacent to the endpoints of~$e$ and embed it
    in the face incident to~$e$ outside of~$C$.  Note that each such
    vertex~$v_e$ is a trivial petal of size~2, and it is hence
    realizable.  Clearly, after performing this operation for each
    edge~$e$ of~$C$, every face incident to an edge of~$C$ is a
    triangle.  Note that, besides the new petals of size~2, we do not
    create or modify any other petals.  Moreover, we introduce no
    chords of~$C$ since any new edge is incident to a new vertex.

    For the second step, we traverse all inner faces incident to
    vertices of~$C$ that are not triangles.  Let~$f$ be such a face
    that is incident to some vertex~$v$ of~$C$.  Let~$vu$ and~$vw$ be
    the two edges incident to~$v$ that bound~$f$.  Note that~$vu$
    and~$vw$ do not belong to~$C$ since, according to step~1,~$f$
    would be a triangle in this case.  Hence, since there are no
    outer chords of~$C$,~$u$ and~$w$ are not in~$C$.  We create a new
    vertex~$v_f$, connect it to~$u,v$ and~$w$, and embed it in~$f$.
    This reduces the number of non-triangular faces incident to~$v$ by
    1.  Note that this operation neither produces a petal nor a chord
    of~$f$.  After treating all faces incident to vertices of~$f$ in
    this way, all inner faces incident to vertices of~$f$ are
    triangles.
  
    In the third step, we triangulate the remaining faces arbitrarily.
    None of the edges added in this step is incident to a vertex
    of~$C$.  Hence, this produces neither petals nor chords.  Finally,
    if~$C$ is not strictly internal in the resulting graph, there is
    at most one vertex of~$C$ on the outer face.  To see this, observe
    that the outer face is a triangle, and if two of its vertices
    belong to~$C$, the edge between them would be an outer chord,
    contradicting the properties of our construction.  If this case
    occurs, we add an additional vertex into the outer face and
    connect it to all three vertices of the outer face.
    Afterward~$C$ is strictly internal.  Again this can neither
    produce a petal nor a cycle.

    The resulting graph~$G'$ is triangulated and all its petals are
    realizable.  Obviously, the procedure adds only~$O(n)$ vertices.
    This concludes the proof. \qed
  \end{proof}

In the following we assume that our given plane graph is 
maximal, and the given cycle is strictly internal, otherwise 
Proposition~\ref{prop:triangulation} is applied.

\begin{lemma}
  \label{lem:petal-cycle}
  Let~$G$ be a triangulated planar graph with a strictly internal outerchordless
  cycle~$C$.  Then the following statements hold.
  $(i)$ Each vertex of~$C$ that is not internal to an outermost petal is
    adjacent to two outermost petals. $(ii)$ There is a simple cycle~$C'$ whose interior contains~$C$ and
    that contains exactly the outermost stamens and petals of $C$.
\end{lemma}

\begin{proof}
  For~(i), observe that each edge of~$C$ is incident to a triangle
  outside of~$C$, whose tip then is a petal.  Thus, every vertex
  of~$C$ is adjacent to at least two petals~$p_1$ and~$p_2$ such that
  none of them contains the other one.  For a vertex~$v$ of~$C$ that
  is not contained in the base of an outermost petal, this implies
  that there exists distinct outermost petals~$p_1'$ and~$p_2'$
  with~$p_1 \prec p_1'$ and~$p_2 \prec p_2'$.  Since~$v$ is not
  internal to any of them, it must be incident to one of their outer
  edges.

  For~(ii), let~$v$ be a vertex of~$C$ and let~$u$ and~$w$ be two
  stamens or petals that are adjacent in the circular ordering of
  neighbors around~$v$.  Then~$vu$ and~$vw$ together bound a face,
  which must be a triangle, implying that~$u$ and~$w$ are adjacent.
  Applying this argument to all adjacent pairs of petals or stamens
  around~$C$ yields the claimed cycle in~$G$ that traverses exactly
  the outermost petals and stamens.
  \qed
\end{proof}

The following lemma is a helper lemma that is only used in the proof
of Lemma~\ref{lem:disjoint-paths}.

\begin{lemma}
  \label{lem:path-in-petal}
  Let~$G$ be a triangulated planar graph with a strictly internal outerchordless
  cycle~$C$.  Let~$u \in C$ be internal to an
  outermost petal~$v$.  Then there exists a chordless path from~$u$
  to~$v$ that contains no other vertices of~$C$.
\end{lemma}

\begin{proof}
  If~$u$ is adjacent to~$v$, there is nothing to prove.  Otherwise,
  let~$v_1$ and~$v_2$ be the first vertex on~$C$ to the left and right
  of~$u$, respectively, that are adjacent to~$v$.  Let~$C'$ denote the
  cycle consisting of~$C[v_1,v_2]$ and~$v$ together with two
  edges~$vv_1$ and~$vv_2$.  The graph consisting of~$C'$ and all edges
  and vertices embedded inside~$C'$ is inner-triangulated and
  chordless, and hence triconnected.  Thus, in this graph, there exist
  three vertex disjoint paths from~$u$ to~$v$.  The middle one of
  these paths cannot contain any vertices of~$C$.  We obtain the
  claimed path by removing transitive edges from the middle path.
  \qed
\end{proof}

\begin{lemma}
  \label{lem:disjoint-paths}
  Let~$G$ be a maximal planar graph with a strictly internal outerchordless
  cycle~$C$.  Let~$u$ and~$v$ be two adjacent
  vertices on~$C$ that are not internal to the same petal.
	Then there exists a third vertex~$w$ of~$C$ such that there exist three
  chordless disjoint paths from~$u,v$ and~$w$ to the vertices of the
  outer face of~$G$ such that none of them contains other vertices
  of~$C$.
\end{lemma}

\begin{proof}
  Let~$C_{\shell}$ denote the cycle going through the petals and
  stamens of~$C$, which exists by Lemma~\ref{lem:petal-cycle}.  It
  follows from Menger's theorem that there exists a set~$A$ of three
  vertices on~$C_{\shell}$ that have disjoint paths to the outer face,
  such that none of them contains any other vertices of~$C_{\shell}$.
  Our goal is to connect the vertices~$u,v,w$ by vertex-disjoint paths
  to these vertices such that together we obtain the claimed paths.
  Assume without loss of generality that~$v$ is the immediate
  successor of~$u$ on~$C_{\shell}$ in counterclockwise order.

  We distinguish cases based on the number of outermost petals.

  \textbf{Case 1:} There are exactly two outermost petals~$x$ and~$y$.
  Note that one of~$u$ and~$v$ must be adjacent to both.  Without loss
  of generality, we assume that~$v$ is.  The case that~$u$ is adjacent
  to both of them is completely symmetric.

  We consider the subpaths of~$C_{\shell}$ from~$x$ to~$y$.  At least
  one of them contains a vertex of~$A$ in its interior.  If one of
  them contains all three of them in its interior, let~$b$ denote the
  middle one, otherwise, let~$b$ denote an arbitrary one.  Let~$a$
  and~$c$ denote the vertices of~$A \setminus \{b\}$ that occur before
  and after~$b$ in counterclockwise direction along~$C_{\shell}$.

  We distinguish cases based on whether~$b$ is contained in the
  clockwise or in the counterclockwise path from~$x$ to~$y$.

  \textbf{Case 1a:} Vertex~$b$ is contained in the counterclockwise
  path from~$x$ to~$y$.  It then follows that~$b$ is a stamen
  of~$v$. We choose~$w$ as an arbitrary vertex on the base of
  petal~$y$ such that it is distinct from~$u$ and~$v$ (note that this
  is always possible since~$C$ has length at least~3).

  Now the claimed paths can be constructed as follows.  For~$u$ take
  the path from~$u$ to~$x$ (Lemma~\ref{lem:petal-cycle}), and from
  there along~$C_{\shell}$ to~$a$, avoiding~$b$.  For~$v$, take the
  path~$vb$.  For~$w$ take the path from~$w$ to~$y$
  (Lemma~\ref{lem:petal-cycle}), and from there along~$C_{\shell}$
  to~$c$, avoiding~$b$.  It follows from the choice of~$b$ that these
  paths are disjoint.

  \textbf{Case 1b:} Vertex~$b$ is contained in the clockwise path
  from~$x$ to~$y$.  Then~$b$ must be a stamen that is, by
  construction, not adjacent to~$v$.

  If~$b$ is not adjacent to~$u$, choose~$b$'s unique neighbor on~$C$
  as~$w$.  The paths are constructed as follows.  For~$u$, take the
  path from~$u$ to~$x$ and from there to~$c$, avoiding~$b$.  For~$v$,
  take the path from~$v$ to~$y$ and from there to~$a$, avoiding~$b$.
  For $w$ take the edge~$wb$.  It follows from the choice of~$b$ that
  the constructed paths are vertex disjoint.

  Otherwise~$b$ is adjacent to~$u$.  Then choose~$w$ as a vertex
  adjacent to~$y$ but distinct from~$u$ and~$v$ (such a vertex exists
  since~$C$ has length at least~3).  The paths are constructed as
  follows.  For~$u$, we take the edge~$ub$.  For~$v$, we take the path
  from~$v$ to~$x$ (Lemma~\ref{lem:path-in-petal}) and from there
  to~$c$, avoiding~$b$.  For~$w$, we take the path from~$w$ to~$y$
  (Lemma~\ref{lem:path-in-petal}) and from there to~$a$, avoiding~$b$.
  Again, it follows from the choice of~$b$ that the constructed paths
  are vertex disjoint.

  \textbf{Case 2:} There are at least three petals of~$C$.  Let~$x$
  denote an outermost petal adjacent to~$u$ such that an outer edge
  of~$x$ is incident to~$v$.  Let~$y \ne x$ be an outermost petal
  adjacent to~$v$.  One of the two subpaths of~$C_{\shell}$
  connecting~$x$ and $y$ contains in its interior a vertex of~$A$.  If
  there are three vertices in the interior of one of these paths,
  let~$b$ denote the middle one.  The vertices of~$A$ before and after
  it in counterclockwise direction along~$C_{\shell}$ are~$a$ and~$c$,
  respectively.

  We distinguish cases based on whether~$b$ is contained in the
  clockwise or in the counterclockwise path from~$x$ to~$y$
  along~$C_{\shell}$.

  \textbf{Case 2a:} Vertex~$b$ is in the counterclockwise path
  from~$x$ to~$y$ along~$C_{\shell}$.

  Then,~$b$ is a stamen adjacent to~$v$ by the choice of~$x$.  Let~$w$
  be any vertex adjacent to~$y$ but distinct from~$v$.  Note that~$w
  \ne u$ by the assumption that there are at least three petals.  We
  take construct the following paths.  For~$u$, we apply
  Lemma~\ref{lem:path-in-petal} to construct a path from~$u$ to~$x$
  and then follow~$C_{\shell}$ to~$a$, avoiding~$b$.  For~$v$, we take
  the edge~$vb$.  Finally, for~$w$, we take the edge~$wy$ and then
  traverse~$C_{\shell}$ to~$c$, avoiding~$b$.  Note that the paths are
  disjoint by the choice of~$b$.

  \textbf{Case 2b:} Vertex~$b$ is in the clockwise path from~$x$
  to~$y$ along~$C_{\shell}$.  Then,~$b$ is a petal or stamen but, by
  definition, not adjacent to~$v$.  If~$b$ is adjacent to a vertex
  that is distinct from~$u$ as well, then we choose this vertex
  as~$w$.  The paths are then~$wb$, and we use the path from~$u$
  to~$x$ to~$a$ (along $C_{\shell}$) and the path from~$v$ to~$y$
  to~$b$ as above.

  If, however~$b$ is a stamen of~$v$, then pick~$w$ as a vertex
  adjacent to~$y$ but distinct from~$u$ and~$v$ (such a vertex exists
  since petal~$y$ has at least two neighbors on~$C$, and since there
  are at least three petals).  Then the paths are~$ub$, the path
  from~$v$ to~$x$ to~$c$ and from~$w$ to~$y$ to~$a$.  As above the
  paths can be chosen such that they are disjoint.

  In all cases, by taking the transitive reduction of the constructed
  paths, we obtain the claimed chordless paths.
  \qed
\end{proof}

\section{Extension of a one-sided polygon}
\label{sec:extension-one-sided}

Let $G$ be a plane graph, and let $C$ be a simple outerchordless cycle, represented by a one-sided polygon $\Gamma_C$. In this section, we show that if each petal of $C$ is realizable, then $\Gamma_C$ 
is extendable to a straight-line drawing of $G$. This result serves as a tool for the general case, which is shown in Section~\ref{sec:main-theorem}.

The drawing extension we produce preserves the outer face, i.e., if the extension exists, then it has the same outer
face as $G$. It is worth mentioning that, if we are allowed to change the outer face, the proof becomes rather simple, as the following claim shows.

\begin{remark} 
\label{theor:one-sided-nochange} Let $G$ be a maximal plane graph and let $C$ be an outerchordless cycle of $G$, represented in the plane by a one-sided polygon $\Gamma_C$. Then drawing $\Gamma_C$ is extendable.
\end{remark}
\begin{proof}
Let $(v_1,v_k)$ be the base edge of $\Gamma_C$.
Edge $(v_1,v_k)$ is incident to two faces of $G$, to a face $f_{\inter}$ inside $C$ and to a face $f_{\out}$ outside $C$. 
We select $f_{\out}$ as the
outer face of $G$. With this choice, edge
$(v_1,v_k)$ is on the outer face of $G$. Let $v$ be the third vertex of this face. 
We place the vertex $v$ far enough from $\Gamma_C$, so that all vertices of $\Gamma_C$ are visible from $v$. Thus, we obtain a planar straight-line drawing of the subgraph $G_v$ induced by the
vertices $V(C) \cup \{v\}$ such that each face is represented by a star-shaped polygon. Each subgraph of $G$ inside a face of $G_v$ is triconnected, and therefore, we can complete the existing drawing to a straight-line planar drawing of $G$, by Theorem~\ref{theorem:HongNagamochi}. \qed
\end{proof}

In the rest of the section we show that extendability of $\Gamma_C$ can be efficiently tested, even if the outer face of $G$ has to be preserved. The following simple geometric fact will be used in the proof of the result of this section (see Figure~\ref{fig:geom_stat} for the illustration). 
\begin{fact}
\label{stat:no_crossings}
Let $pqrt$ be a convex quadrilateral and let $o$ be the intersection of its diagonals. Let $S_{pt}$ be a one-sided convex polygon with base $\overline{pt}$, that lies inside triangle $\triangle opt$. Let $\overline{ab}$ and $\overline{cd}$ be such that $b,d \in S_{pt}$, ordered clockwise as $t,d,b,p$ and $a,c \in \overline{qr}$, ordered as $q,a,c,r$. Then, neither  $\overline{ab}$ and $\overline{cd}$  intersect each other, nor do they intersect a segment between two consecutive points of $S_{pt}$. 
\end{fact}

We are now ready to prove the main result of this section.
\begin{theorem}
\label{theor:one-sided} Let $G$ be a plane graph and $C$ be a simple outerchordless cycle of $G$, represented in the plane by a one-sided polygon $\Gamma_C$.
If every petal of $C$ is realizable, then $\Gamma_C$ is extendable.
\end{theorem}
\begin{proof}
For the following construction we need to ensure that $G$ is a maximal plane graph and cycle $C$ is strictly internal. If this is not the case, we apply Proposition~\ref{prop:triangulation}, to complete $G$ to a maximal plane graph, where $C$ is strictly internal. For simplicity of notation, we refer to the new graph as to $G$. Notice that the application of Proposition~\ref{prop:triangulation} yields introduction only of trivial, and therefore, realizable petals. Moreover, no outer chord is introduced. After the construction of the extension of $C$, we remove the vertices and edges added by the application of Proposition~\ref{prop:triangulation}.

Let $v_1,\dots,v_k$ be the clockwise ordering of the vertices of $C$, so that $(v_1,v_k)$ is the base of $\Gamma_C$. We rotate $\Gamma_C$ so that $(v_1,v_k)$ is horizontal.  Let $a,b,c$ be the vertices of the external face of $G$, in clockwise order, see Fig.~\ref{fig:menger}. By Lemma~\ref{lem:disjoint-paths}, there exists a vertex $v_j$, $1<j<k$, such that there exist chordless disjoint paths between $v_1,v_j,v_k$, and the vertices $a,b,c$, respectively. Without loss of generality assume they are $P_{v_1a}$, $P_{v_jb}$ and $P_{v_kc}$. Some vertices of $P_{v_1a}$ and $P_{v_kc}$ are possibly adjacent to each other, as well as to the boundary of $C$. Depending on these adjacencies, we show how to draw the paths $P_{v_1a}$, $P_{v_kc}$ and how to place vertex $b$, so that the graph induced by these vertices and cycle $C$ is drawn with star-shaped faces. Then, the drawing of $G$ can be completed by applying Theorem~\ref{theorem:HongNagamochi}.
Let $v_i$ be the topmost vertex of $\Gamma_C$. It can happen that there are two adjacent topmost vertices $v_i$ and $v_{i+1}$.
However, $v_{i-1}$ and $v_{i+2}$ are lower, since $\Gamma_C$ does not contain flat vertices. In the following, we assume that $v_{i}$ and $v_{i+1}$ have the same $y$-coordinate. The case when $v_i$ is unique can be seen as a special case where $v_i=v_{i+1}$.
Without loss of generality assume that $i+1\leq j \leq k-1$, the case where $2\leq j \leq i$ is treated symmetrically.
Notice that the presence of the path $P_{v_jb}$ ensures that edges between vertices of $P_{v_1a}$ and $P_{v_kc}$ can only lie in the interior of the cycle delimited by these paths and  edges $(v_1,v_k)$ and $(a,c)$ (refer to Figure~\ref{fig:menger}). Consider a vertex of $P_{v_1a}$ which is a petal of $C$.  The base of such a petal cannot contain edge $(v_{k-1},v_k)$, since this would cause a crossing with $P_{v_kc}$.  Moreover, if the base of this petal contains edge $(v_1,v_k)$, then it cannot contain any edge $(v_f,v_{f+1})$ for $i\le f < j$, since otherwise this petal would not be realizable. Thus a vertex of $P_{v_1a}$ is either adjacent  to $v_k$ or to a vertex $v_f$, where $i+1 \leq f \leq j$, but not both. It is worth mentioning that a vertex of $P_{v_1a}$ cannot be adjacent to any $v_f$, $j+1 \leq f\leq k-1$, since such an adjacency would cause a crossing either with $P_{v_jb}$ or with $P_{v_kc}$. 
\begin{figure}[t]
 \centering
 {\includegraphics[scale=0.55]{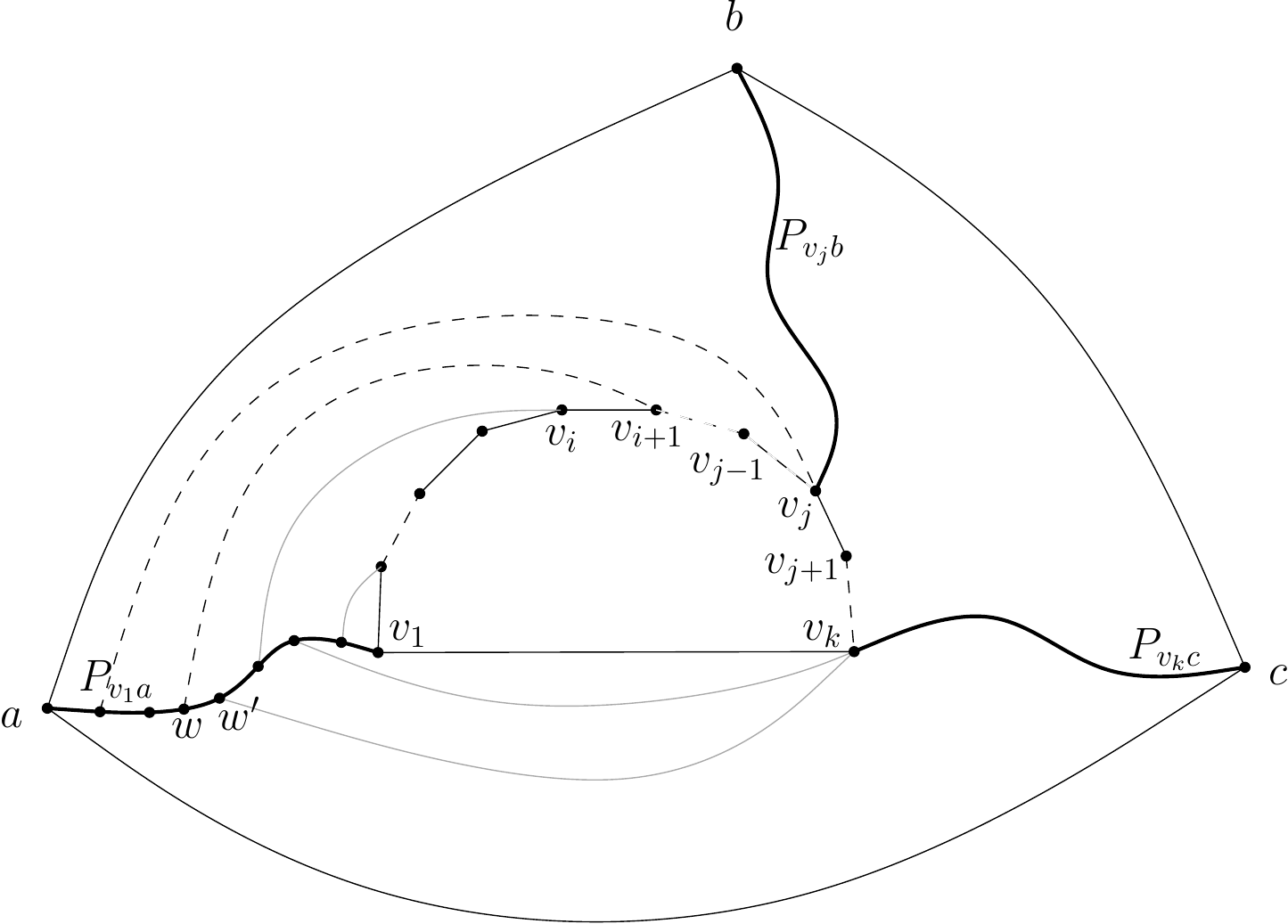}}
 \caption{Illustration for the proof of Theorem~\ref{theor:one-sided}. Edges between $C[v_1,v_i]$ and $P_{v_1a}[v_1,w']\cup \{v_k\}$ are gray. Edges between $C[v_{i+1},v_j]$ and $P_{v_1a}[w,a]$ are dashed.}
 \label{fig:menger} 
\end{figure} 

Let $\ell_a$, $\ell$ and $\ell_c$ be three distinct lines through $v_j$ that lie clockwise between the slopes of edges $(v_{j-1},v_j)$ and $(v_j,v_{j+1})$ (see Figure~\ref{fig:one_sided}). Such lines exist since $\Gamma_C$ does not contain flat vertices. 
Let $\ell_i$ be the line through $v_i$ with the slope of $(v_{i-1},v_i)$. Let $\ell_a^1$ be the half-line originating at an internal point of $(v_1,v_k)$ towards $-\infty$, slightly rotated counterclockwise from the horizontal position, so that it crosses $\ell_i$. Let $q$ denote the intersection point of $\ell_a^1$ and $\ell_i$. Let $p$ be any point on $\ell_a^1$ further away from $v_1$ than $q$. Let $\ell_a^2$ be the line through $p$ with the slope of $\ell$. By construction of lines $\ell_a$, $\ell$ and $\ell_c$,  line $\ell_a^2$ crosses $\ell_a$ above the polygon $\Gamma_C$ at point $p_a$ and line $\ell_c$ below this polygon at point $p_c$.  

Let $G'$ be the plane subgraph of $G$ induced by the vertices of $C$, $P_{v_1a}$, and $P_{v_kc}$. The outer cycle of $G'$ consists of edge $(a,c)$ and a path $P_{ac}$ between vertices $a$ and $c$.
%Next we show the following
\begin{stat}
The vertices of $P_{v_1a}$ and $P_{v_kc}$ can be placed on lines $\ell_a^1$, $\ell_a^2$ and  $\ell_c$ such that in the resulting straight-line drawing of $G'$, path $P_{ac}$ is represented by an $x$-monotone polygonal chain, and the inner faces of $G'$ are star-shaped polygons.
\end{stat}
The vertices of $P_{v_1a}$ will be placed on line $\ell_a^1$ between points $p$ and $q$ and on line $\ell_a^2$ above point $p_a$. The vertices of $P_{v_kc}$ will be placed on $\ell_c$ below $p_c$.
In order to place the vertices, we need to understand how the vertices of $P_{v_1a}$ are adjacent to vertices of $C$.
As we travel on $P_{v_1a}$ from $v_1$ to $a$, we first meet all vertices adjacent to $v_1,\dots,v_i$ and then all vertices adjacent to $v_{i+1},\dots,v_j$, since $G$ is a planar graph. Let $w$ be the first vertex of $P_{v_1a}$ adjacent to $v_{f}$, $i+1 \leq f \leq j$, and let $w'$ be the vertex preceding $w$ on $P_{v_1a}$. We place vertices of $P_{v_1a}[v_1,w']$, in the order they appear in the path, on line $\ell_a^1$, between $q$ and $p$, in increasing distance from $v_1$. We place all vertices of $P_{v_1a}[w,a]$ on $\ell_a^2$ above $p_a$ in increasing distance from $p$. We draw the edges between the vertices of $C$ and $P_{v_1a}$.  Notice that vertex $w$ might not exist, since it might happen that none of the vertices of $P_{v_1a}$ is adjacent to $v_{f}$, $i+1 \leq f \leq k$. In this case all vertices of $P_{v_1a}$ are placed on line $\ell_a^1$, between $q$ and $p$.
\begin{figure}[t]
 \centering
 {\includegraphics[scale=0.55]{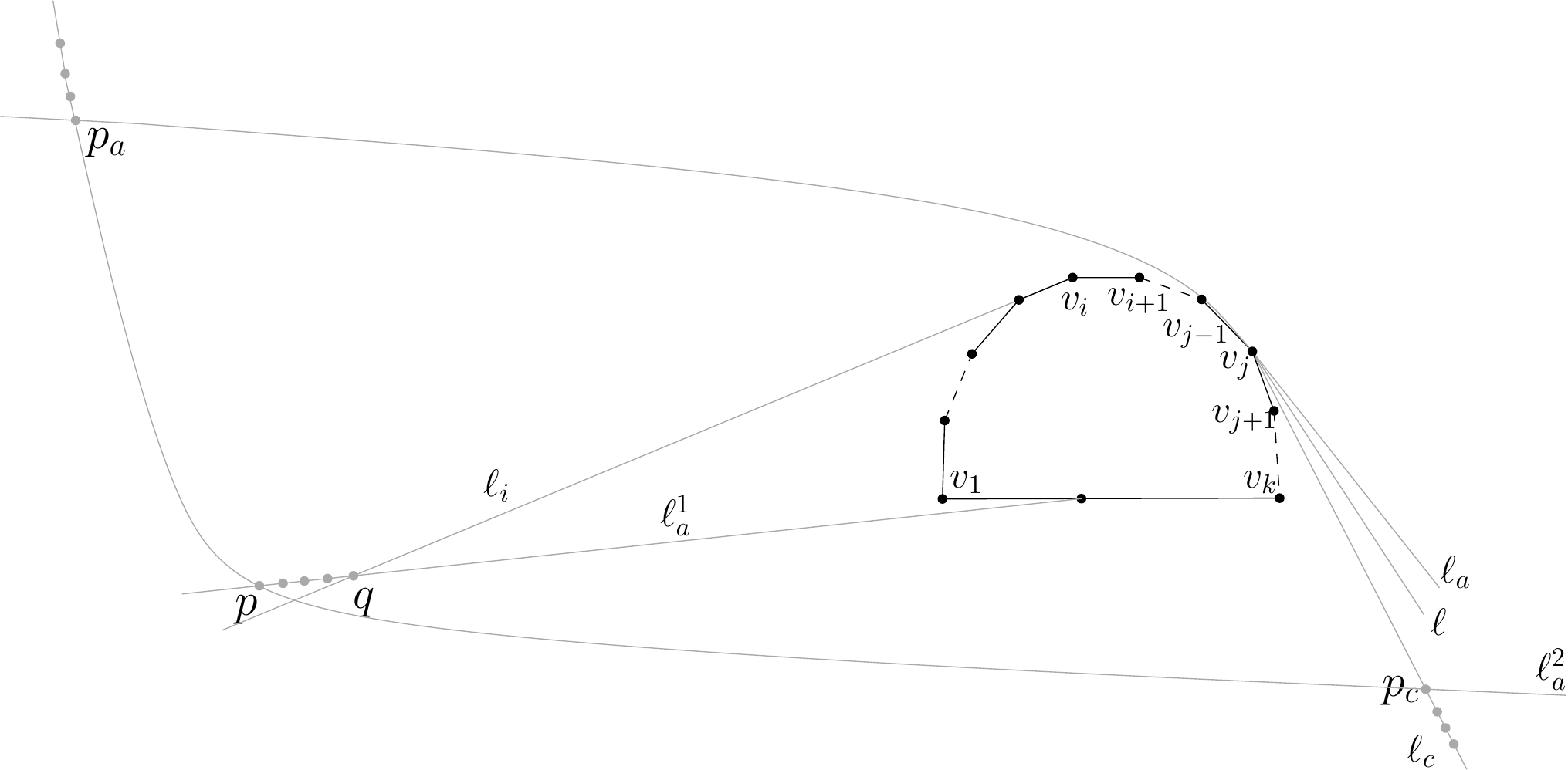}}
 \caption{Illustration for the proof of Theorem~\ref{theor:one-sided}. For space reasons lines were shown by curves.}
 \label{fig:one_sided} 
\end{figure}  
In the following, we show that the constructed drawing is planar. Notice that the quadrilateral formed by the points $w,a,v_j,v_{i+1}$ is convex, by the choice of line $\ell_a^2$ and the positions of vertices $w$ and $a$ on it. Also, notice that the points of vertices $v_{i+1},\dots,v_j$ form a one-sided polygon with base segment $\overline{v_{i+1}v_j}$, which lies in the triangle $\triangle ov_jv_{i+1}$, where $o$ is the intersection of $\overline{v_{i+1}a}$
and $\overline{v_{j}w}$. Thus, by Fact~\ref{stat:no_crossings}, the edges connecting $C[v_{i+1},v_j]$ and $P_{v_1a}[w,a]$  do not cross each other. By applying Fact~\ref{stat:no_crossings}, we can also prove that edges connecting $P_{v_1a}[v_1,w']$ with $C[v_1,v_i]$, cross neither each other, nor $\Gamma_C$. Recalling that vertices of $P_{v_1a}[v_1,w']$ can be also adjacent to $v_k$, we notice that these edges also do not cross $\Gamma_C$, by the choice of line $\ell_a^1$. Finally, path $P_{v_1a}$ is chordless, and therefore the current drawing is planar. Notice that the subpath of $P_{a,c}$ that has already been drawn is represented by an $x$-monotone chain.
We next draw the vertices of $P_{v_kc}$. We observe that in the already constructed drawing path $P_{v_1a}$ taken together with edge $(v_1,v_k)$ is represented by an $x$-monotone chain, each point of which is visible from any point below the line $\ell_a^2$. This means that any point below line $\ell_a^2$, can be connected by a straight-line segment to the vertices $V(P_{v_1a}) \cup \{v_k\}$ without creating any crossing either with $P_{v_1a}$ or with $(v_1,v_k)$.  
We also notice that any of the vertices $v_j,\dots,v_k$ can be connected to a point of $\ell_c$, without intersecting $\Gamma_C$.
Recall that $p_c$ denotes the intersection point of $\ell_c$ and $\ell_a^2$. Thus we place the vertices of $P_{v_kc}$ on the line $\ell_c$, below $\ell_a^2$, in  increasing distance from point $p_c$. Applying Fact~\ref{stat:no_crossings} we can prove that the edges induced by $\{v_j,\dots,v_k\} \cup V(P_{v_kc})$ are drawn without crossings. Edges between $P_{v_kc}$ and $P_{v_1a}$ cross neither $P_{v_1a}$, nor $(v_1,v_k)$ by the choice of lines $\ell_c$ and $\ell_a^2$. 

We have constructed a planar straight-line drawing of $G'$. We notice that path $P_{ac}$ is drawn as an $x$-monotone polygonal chain. We also notice that the faces of $G'$, created when placing vertices of $P_{v_1a}$ (resp. $P_{v_kc}$) are star-shaped and have their kernels arbitrarily close to the vertices of  $P_{v_1a}$ (resp. $P_{v_kc}$). 
%This concludes the proof of the claim.

Notice that vertex $b$ is possibly adjacent to some of the vertices of $P_{ac}$. Thus, placing $b$ at an appropriate distance above $P_{ac}$, the edges between $b$ and $P_{ac}$ can be drawn straight-line without intersecting $P_{ac}$    and therefore no other edge of $G'$. The faces created when placing $b$ are star-shaped and have their kernels arbitrarily close to $b$. We finally apply Theorem~\ref{theorem:HongNagamochi}. \qed 
\end{proof}

\section{Main theorem}
\label{sec:main-theorem}

Let $G$ be a maximal plane graph and $C$ be a strictly internal simple outerchordless cycle of $G$, represented by an arbitrary convex polygon $\Gamma_C$ in the plane.
In Theorem~\ref{theorem:convex} we prove that it is still true that if each petal of $C$ is realizable, then $\Gamma_C$ is extendable.
Before stating and proving Theorem~\ref{theorem:convex}, we introduce notation that will be used through this section. 

Recall that $v_1,\dots,v_k$ denote the vertices of $C$. Let $w_{i,j}$ be an outermost petal of $C$ in $G$. 
Let $\ell_i$ (resp. $\ell_j$) be a half-line with the slope of edge $(v_i,v_{i+1})$ (resp. $(v_{j-1},v_j)$) originating at $v_i$ (resp. $v_j$)
 (see
Figure~\ref{fig:cone_apex}). Since $w_{i,j}$ is realizable, 
lines $\ell_i$ and $\ell_j$ intersect. Denote by $\apex(w_{i,j})$ their
intersection point and by $\cone(w_{i,j})$ the subset of $\mathbb{R}^2$ that
is obtained by the intersection of the half-planes defined by $\ell_i$ and $\ell_j$, not containing 
$\Gamma_C$. It is clear that any internal point of 
$\cone(w_{i,j})$ is appropriate to draw $w_{i,j}$ so that the plane subgraph of $G$
induced by $V(C) \cup \{w_{i,j}\}$ is crossing-free. For consistency, we also define $\cone(w)$ and $\apex(w)$
 of an outer stamen $w$ of $C$ as follows. Assume that $w$ is adjacent to $v_i \in C$. Then $\cone(w) \subset \mathbb{R}^2$ is the union of the half-planes defined by lines of edges $(v_{i-1},v_i)$ and $(v_i,v_{i+1})$, that do not contain $\Gamma_C$. We set $\apex(w)=v_i$.  

\begin{figure}[t]
 \centering
 \subfigure[\label{fig:cone_apex}]
 {\includegraphics[scale=0.7]{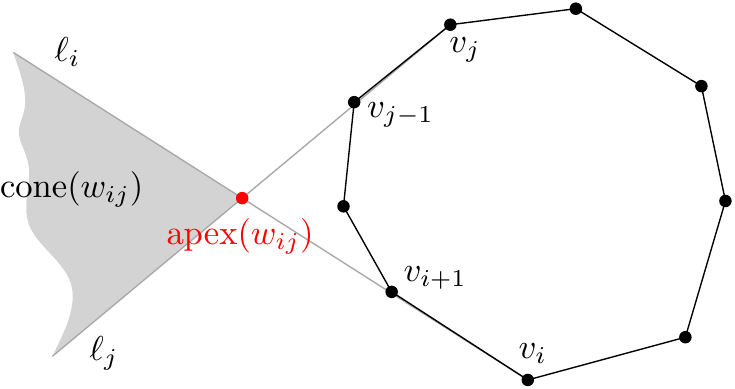}}
 \hspace{+1cm}
 \subfigure[\label{fig:main_theorem}]
 {\includegraphics[scale=0.55]{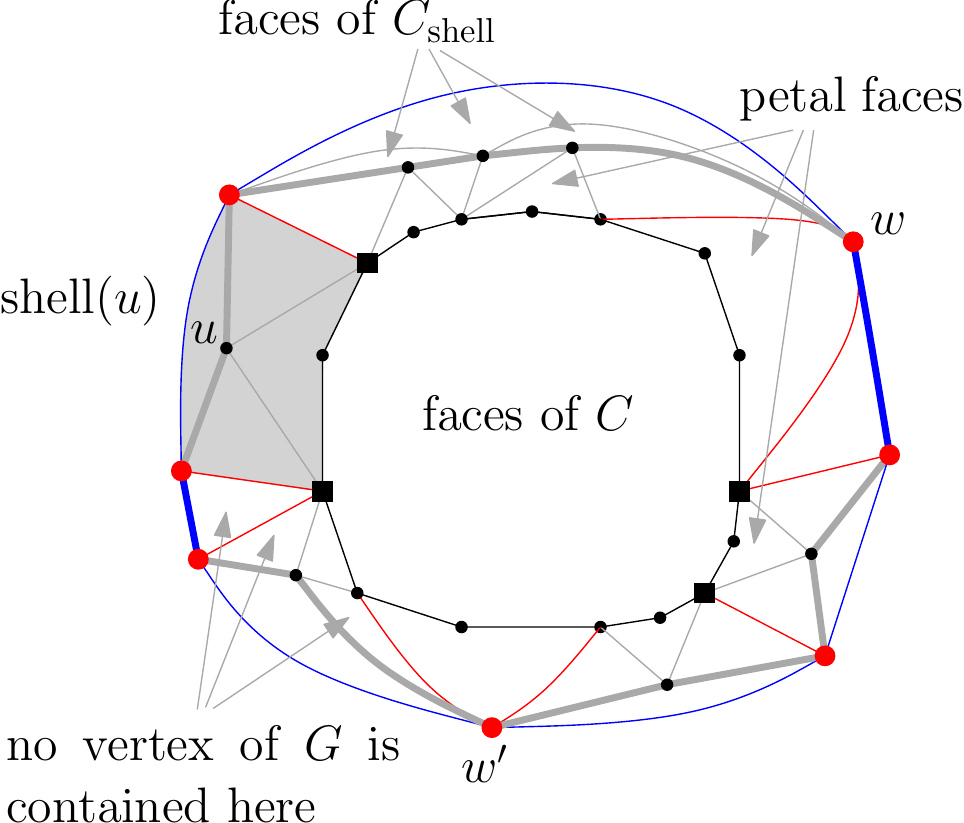}}
 \caption{(a) Vertex $w_{i,j}$ is the petal of $C$ with base $C[v_i,v_j]$. Point $\apex(w_{i,j})$ is red, region $\cone(w_{i,j})$ is gray. (b) Graph $G_{\shell}$. Polygon $\Gamma_C$ is black. Cycle $C_{\shell}$ is bold gray. Cycle $C_{\shell}'$ is blue. Graph $G_{\shell}'$ is comprised by blue, red and black edges. Vertices of $B$ are squares. }
\end{figure}

Let $P$ (resp. $S$) denote the set of outermost petals (resp. stamens) of $C$ in $G$. 
By Lemma~\ref{lem:petal-cycle}, there exists a cycle~$C_{\shell}$ in~$G$ that
contains exactly~$P \cup S$. Let $G_{\shell}$ denote
the plane subgraph of $G$ induced by the vertices of $C$ and $C_{\shell}$. (Figure~\ref{fig:main_theorem}). Let $C_{\shell}'$ denote the outer cycle of $G_{\shell}$. We denote the graph consisting of $C$, $C_{\shell}'$ and edges between them by $G_{\shell}'$. Each petal or stamen of $C$, say $w$, that belongs to $C_{\shell}$ but not to $C_{\shell}'$, belongs to a face of $G_{\shell}'$. We denote this face by $\shell(w)$. 
We categorize the faces of $G_{\shell}$ as follows. The faces that lie inside cycle $C$ are called \emph{faces of $C$}.  The faces that are bounded only by $C_{\shell}$ and its chords, are called \emph{faces of $C_{\shell}$}. Notice that each face of $C_{\shell}$ is a triangle. 
Notice that a face of $G_{\shell}$ that is comprised by two consecutive edges adjacent to the same vertex of $C$ (not belonging to $C$), is a triangle, and contains no vertex of $G \setminus G_{\shell}$, since both facts would imply that the taken edges are not consecutive.
Finally, faces bounded by a subpath of $C$ and two edges adjacent to the same petal, are called \emph{petal faces}. The plane subgraph of $G$ inside a petal face is triangulated and does not have a chord connecting two vertices of its outer face, and therefore is triconnected.  Thus we have the following 

\begin{obs}
\label{obs:G_structure}
Each vertex of $G \setminus G_{\shell}$ either lies in a face of $C$, or in a face that is a triangle, or in a petal face, or outside $C'_{\shell}$. Each subgraph of $G$ inside a petal face is triconnected.
\end{obs}

\begin{figure}[t]
 \centering
 \subfigure[\label{fig:main_theorem_lines}]
 {\includegraphics[scale=0.55]{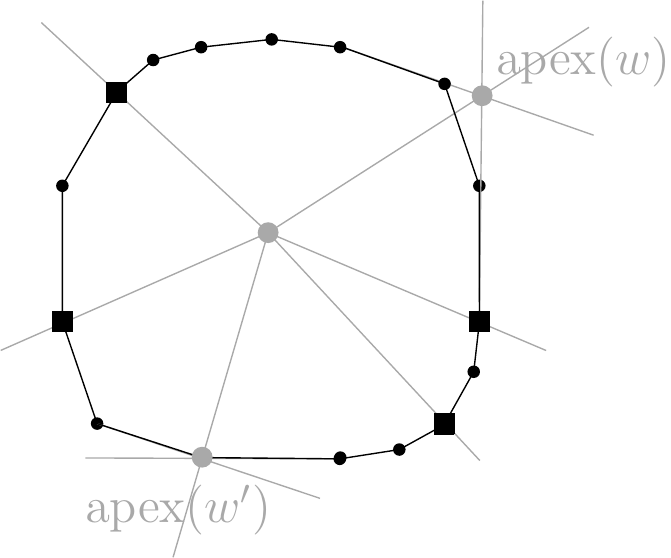}}
 \hspace{+1cm}
 \subfigure[\label{fig:main_theorem_c}]
 {\includegraphics[scale=0.55]{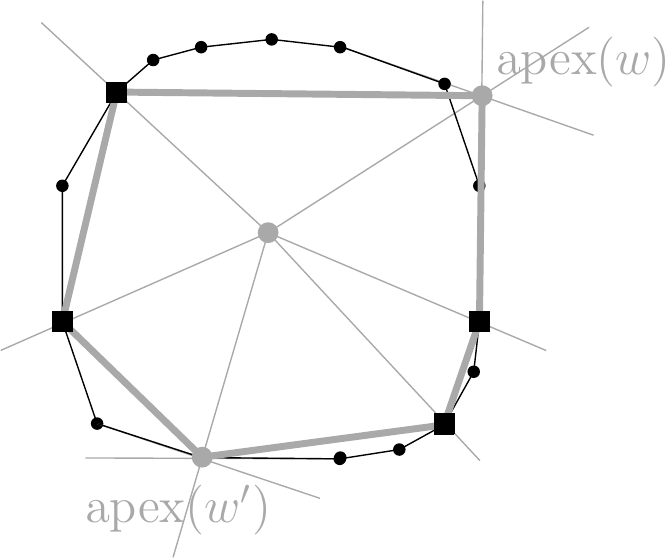}}
 \caption{Construction of drawing of graph $G_{\shell}$ shown in Figure~\ref{fig:main_theorem}. (a) Apex points are gray, points of $B$ are black squares. (b) Polygon $\Pi$ is gray, lines $\{\ell(w) \mid w\in S\ \cap C_{\shell}'\}$ are dashed.}
\end{figure}

\vspace{-0.5cm}
\begin{theorem}
\label{theorem:convex} Let $G$ be a plane graph and let $C$ be
a simple outerchordless cycle of $G$, represented by a convex polygon $\Gamma_C$ in the plane. 
$\Gamma_C$ is extendable to a straight-line drawing of $G$ if and only if each petal of $C$ is realizable.
\end{theorem}

\begin{proof}
Similarly to the proof of Theorem~\ref{theor:one-sided}, we first complete $G$ to a maximal plane graph, such that cycle $C$ becomes strictly internal, by applying Lemma~\ref{prop:triangulation}. The new maximal plane graph (for simplicity of notation denoted also by $G$) contains no outer chord and only realizable petals, since the newly added petals are trivial. When the construction of the extension of $C$ is completed, we simply remove the vertices and edges added by Proposition~\ref{prop:triangulation}. 

The condition that each petal of $C$ is realizable is clearly necessary. Next we show that it is also sufficient.

We first show how to draw the graph $G_{\shell}'$. Afterward we complete it to a drawing of $G_{\shell}$. Our target is to represent $C_{\shell}'$ as a one-sided polygon, so that Theorem~\ref{theor:one-sided} 
can be applied for the rest of $G$ that lies outside $C_{\shell}'$.
We first decide which edge of $C_{\shell}'$ to ``stretch'', i.e., which edge will serve as base edge of the one-sided polygon for representing $C_{\shell}'$. 
In order to be able to apply Theorem~\ref{theor:one-sided}, this one-sided polygon should be such that each petal of $C_{\shell}'$  is realizable. 
Thus we choose the base edge $e$ of $C_{\shell}'$  as follows. If $C_{\shell}'$  contains an edge on the outer face of $G$, we choose $e$ to be this edge. 
Otherwise, we choose an edge $e$, such that at least one of the end vertices of $e$ is adjacent to an outermost petal of $C_{\shell}'$ in $G$. Such a choice of $e$ ensures that each petal of $C_{\shell}'$ is realizable. 
%We show the following 

\begin{stat}
Polygon $\Gamma_C$ can be extended to a straight-line drawing of  graph $G_{\shell}'$, such that its outer face $C_{\shell}'$ is represented by a one-sided polygon with base edge $e$. Moreover, $C_{\shell}'$ contains in its interior all points of $\{\apex(w)\mid w \in C_{\shell} \}$.
\end{stat}

Recall that $P$ (resp. $S$) denotes the set of outermost petals (resp. stamens) of $C$ in $G$.  Let $B$ denote the set of vertices of $C$, to which stamens $S \cap C_{\shell}'$ are adjacent (refer to Figure~\ref{fig:main_theorem}). By construction of the $\apex$ points, the set $\{ \apex(w) \mid w \in P \cap C_{\shell}'\} \cup B$ is in convex position, and we denote by $\Pi$ its convex hull. Polygon $\Pi$ may be degenerate, and may contain only a single vertex or a single edge. We treat these cases separately to complete the construction of the drawing of the graph $G_{\shell}'$.  

\begin{figure}[tb]
 \centering
 {\includegraphics[scale=0.6]{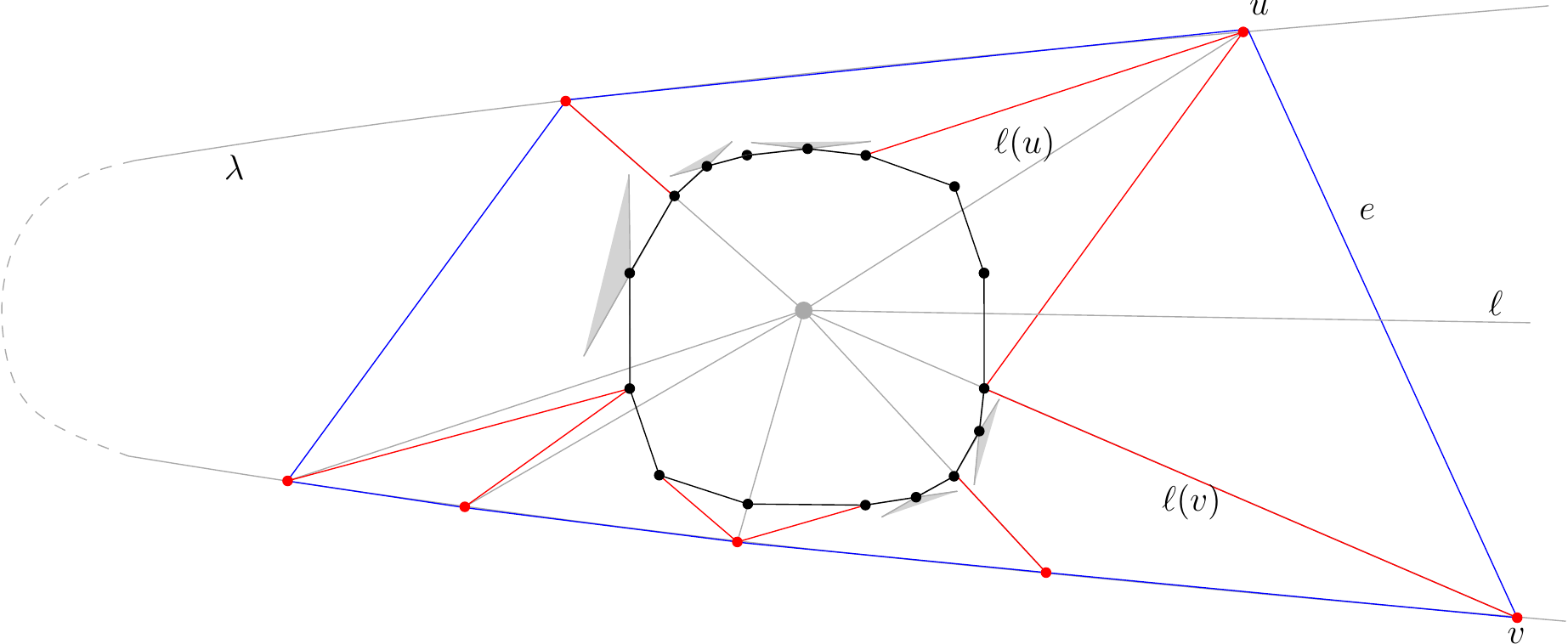}}
 \caption{Construction of Case~1. Corresponding $G_{\shell}$ is shown in Figure~\ref{fig:main_theorem}.}
 \label{fig:main_theorem_lines_d}
\end{figure}

\begin{description}
\item[\bf Case~1: Polygon $\Pi$ is non-degenerate.] 
Let $p$ be a point inside $\Pi$. Let $\ell(w)$ denote a half-line from $p$ through $w$, where $w$ is a vertex of $\Pi$. If we order the constructed half-lines around $p$, any two consecutive lines have between them an angle less than $\pi$. If $w \in B$, we substitute $\ell(w)$ by the same number of slightly rotated lines as the number of stamens of $C_{\shell}'$ adjacent to $w$, without destroying the order (refer to Figure~\ref{fig:main_theorem_c}). Thus, for each $w\in C_{\shell}'$, a line $\ell(w)$ is defined. Notice that, for any $w \in P \cap C_{\shell}'$, line $\ell(w)$ passes through $\apex(w)$, and the infinite part of $\ell(w)$ lies in $\cone(w)$. Thus, for any position of $w$ on a point of $\ell(w) \cap \cone(w)$, edges between $C$ and $w$ do not cross $\Gamma_C$. For a stamen $w \in S \cap C_{\shell}'$, line $\ell(w)$ crosses $\cone(w)$ very close to $\apex(w)$, and its infinite part lies in $\cone(w)$. Thus, similarly, for any position of $w$ on a point of $\ell(w) \cap \cone(w)$, edges between $C$ and $w$ do not cross $\Gamma_C$.
 
Recall that $e=(u,v)$ is the edge of $C'_{\shell}$ that we have decided to ``stretch''. Recall also that $\ell(u)$ and $\ell(v)$ are consecutive in  the sequence of half-lines 
we have constructed. Let $\kappa$ be a circle around $\Gamma_C$ that contains in the interior the polygon $\Pi$ and the set of points $\{\apex(w)| w\in C_{\shell}\}$. Let $\ell$ be a half-line bisecting the angle between $\ell(u)$ and $\ell(v)$ (refer to Figure~\ref{fig:main_theorem_lines_d}). Let $\lambda$ be a parabola with $\ell$ as axis of symmetry and the center of $\kappa$ as focus. We position and parametrize $\lambda$ such that it does not cross $\ell$ and $\kappa$. 

With this placement of $\lambda$, each half-line $\ell(w)$, $w \in \Pi$, crosses $\lambda$, moreover, intersections with $\ell(u)$ and $\ell(v)$ are on different branches of $\lambda$ and appear last on them as we walk on $\lambda$ from its 
origin to infinity. Let $\Pi'$ be the convex polygon comprised by the intersection points of lines $\{\ell(w): w\in V(C_{\shell}')\}$ 
with $\lambda$. We make $\lambda$ large enough, so that the polygon $\Pi'$ still contains the circle $\kappa$ in the interior. As a results, for each $w \in C$, $\cone(w) \cap \Pi'_{\inter} \neq \emptyset$, where $\Pi'_{\inter}$ denotes the interior of $\Pi'$.   
This concludes the proof of the claim in the non-degenerate case.

\item[\bf Case~2: Polygon $\Pi$ is degenerate and has two vertices.] Notice that in this case $C_{\shell}'$ contains at most $2$ petals, since each petal of $C_{\shell}'$ is a vertex of $\Pi$. Assume that $C_{\shell}'$ contains two petals, $w$ and $w'$. Since $C_{\shell}'$ contains at least three vertices ($G$ does not have double edges), there would be at least one stamen $w''$ in $C_{\shell}'$. However, the vertex of $C$, to which $w''$ is adjacent, together with $w$ and $w'$, contribute three vertices to $\Pi$, which is a contradiction to the assumption that $\Pi$ has two vertices. Thus, in this case $C_{\shell}'$ contains at most one petal of $C$.  Thus we consider two subcases, first when $C_{\shell}'$  contains one petal, and second when it contains no petal.

\begin{figure}[tbh]
 \centering
 \subfigure[\label{fig:special_case_2a}]
 {\includegraphics[scale=0.5]{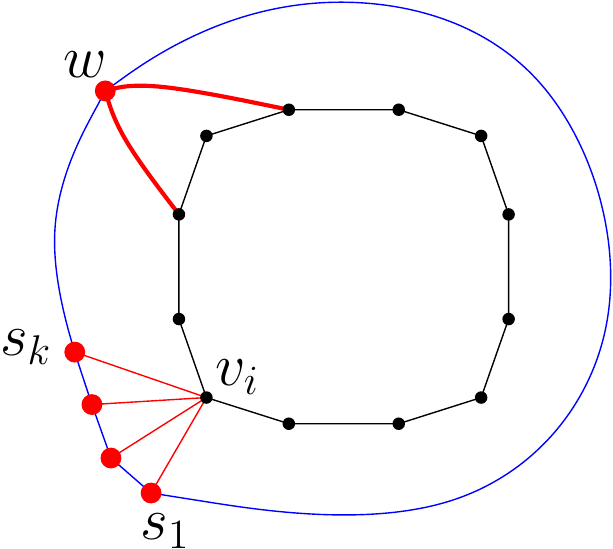}}
 \hspace{+1cm}
 \subfigure[\label{fig:special_case_2b}]
 {\includegraphics[scale=0.5]{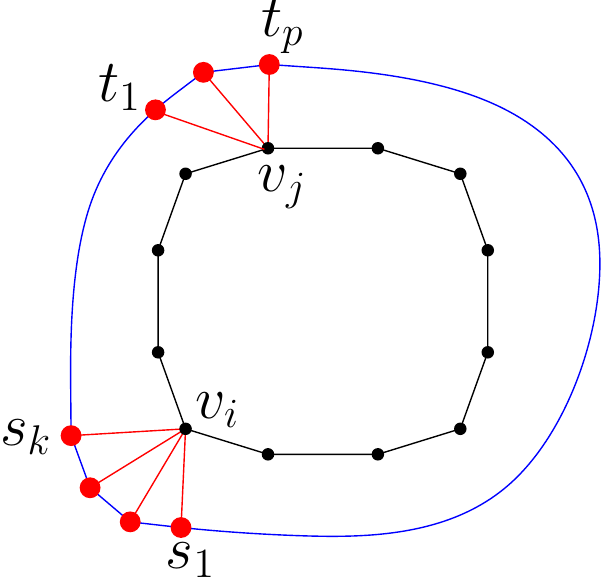}}
 \hspace{+1cm}
 \subfigure[\label{fig:special_case_3b}]
 {\includegraphics[scale=0.5]{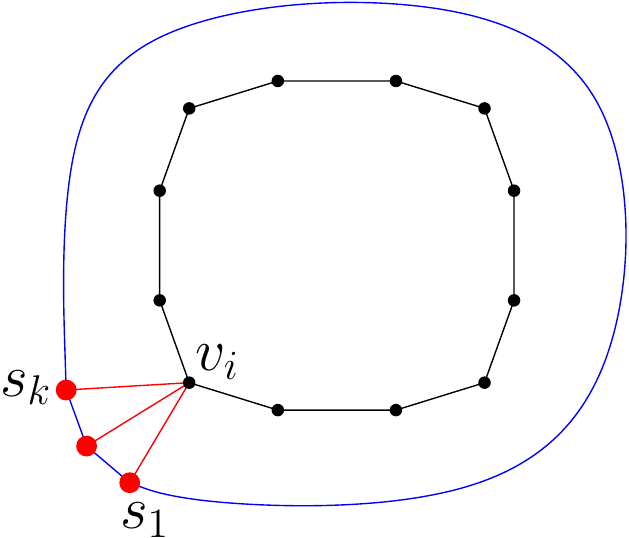}}
 \caption{(a-b) Case~2: Polygon $\Pi$ is degenerated and has two vertices. (c) Case~3: Polygon $\Pi$ is degenerated and has a single vertex, and $C_{\shell}'$ contains three or more stamens of $C$ adjacent to the same vertex.}
\end{figure}

\begin{figure}[tbh]
 \centering
 \subfigure[\label{fig:append_2a}]
 {\includegraphics[scale=0.5]{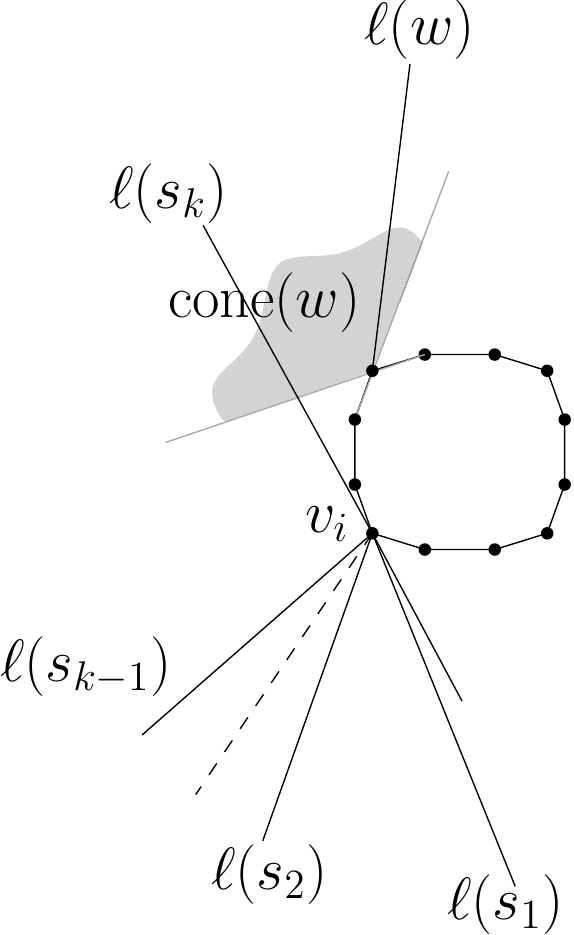}}
 \hspace{+1cm}
 \subfigure[\label{fig:append_3}]
 {\includegraphics[scale=0.5]{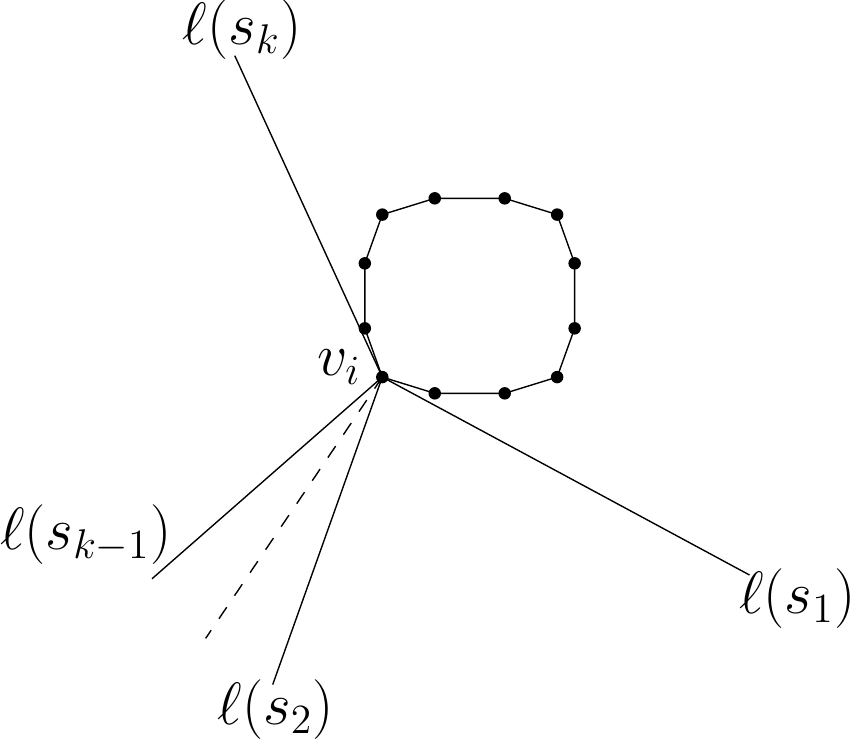}}
 \subfigure[\label{fig:3a}]
 {\includegraphics[scale=0.4]{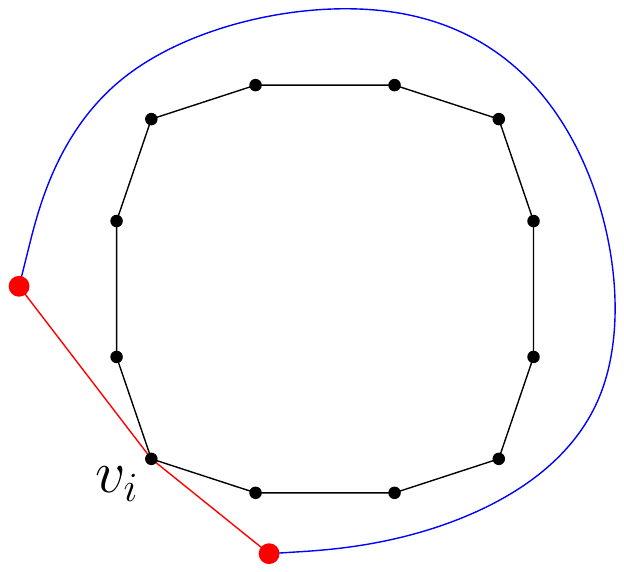}}
 \caption{(a) Construction in Case~2.a. (b) Construction in Case~3. (c) If $C_{\shell}'$  contains only two stamens of $C$, then a vertex of $C$, $v_i$, belongs to the outer cycle of $G$.}
\end{figure}

{\bf Case~2.a: $C_{\shell}'$  contains a petal $w$ of $C$ (Figure~\ref{fig:special_case_2a}).}  
Notice that there are at least two stamens of $C$ in  $C_{\shell}'$, since otherwise  $C_{\shell}'$ would contain only two vertices, and therefore there would be a double edge in $G$. Moreover, all stamens of $C$ that are in $C_{\shell}'$, are adjacent to the same vertex $v_i$ of $C$. Let $s_1,\dots,s_k$ be the stamens adjacent to $v_i$, appearing clockwise in $C_{\shell}'$. Let $\ell(s_k)$ be a line through $v_i$ that does not cross $\Gamma_C$ (Figure~\ref{fig:append_2a}). There exist a line $\ell(w)$ in $\cone(w)$ which does not cross $\ell(s_k)$.
Let $\ell(s_1)$ be a half-line through $v_i$,  lying in the half-plane defined by $\ell(s_k)$ not containing $\Gamma_C$, slightly rotated clockwise. Let $\ell(s_2),\dots,\ell(s_{k-1})$ be lines through $v_i$ that lie clockwise between $\ell(s_1)$ and $\ell(s_k)$. 
The angle between two consecutive half-lines $\ell(s_1),\dots,\ell(s_k), \ell(w)$ is less than $\pi$. Moreover, each of the half-lines has its infinite part in the corresponding cone. Moreover, for any position of $w$ (resp. $s_i$) on a point of $\ell(w) \cap \cone(w)$ (resp. $\ell(s_i) \cap \cone(s_i)$) , edges between $C$ and $w$ (resp. $s_i$) do not cross $\Gamma_C$. Thus, we can apply construction using circle $\kappa$ and parabola $\lambda$ identical to the case when polygon $\Pi$ is not degenerate.    

{\bf Case~2.b: $C_{\shell}'$  does not contain any petal of $C$ (Figure~\ref{fig:special_case_2b}).}  
In this case $C_{\shell}'$ contains at least three stamens of $C$, and all of them are adjacent to two vertices of $C$, say $v_i$ and $v_j$, otherwise $\Pi$ would contain more than two vertices. Without lost of generality assume that at least two stamens of $C$ are adjacent to $v_i$. Let them be $s_1,\dots,s_k$. Let $t_1,\dots,t_p$ be the stamens adjacent to $v_j$. The construction of half-lines $\ell(s_1),\dots, \ell(s_k),\ell(t_1),\dots,\ell(t_p)$ can be done along the same lines as in Case~2.a, where line $\ell(w)$ is substituted by a set of closely placed lines $\ell(t_1),\dots,\ell(t_p)$. Since the angle between two consecutive half-lines $\ell(s_1),\dots, \ell(s_k)$, $\ell(t_1),\dots,\ell(t_p)$  is less than $\pi$, we proceed again as in non-degenerate case.  

\item[\bf Case~3: Polygon $\Pi$ is degenerated and has a single vertex (Figure~\ref{fig:special_case_3b}).] 
We first notice that in this case $C_{\shell}'$ does not contain any petal of $C$. This is because $C_{\shell}'$ contains at least one stamen of $C$, adjacent to a vertex $v_i$ of $C$. Thus a petal together with $v_i$ would contribute two vertices to $\Pi$. Therefore, $C_{\shell}'$  contains only stamens of $C$. These stamens are adjacent to a single vertex of $C$, say $v_i$, since otherwise $\Pi$ would contain more than one vertex. If there are only two stamens adjacent to $v_i$ (see Figure~\ref{fig:3a}), then either there is a double edge in $G$, or $v_i$ belongs to the outer cycle of $G$. The latter can not be the case, since $C$ is a strictly internal cycle of $G$, as ensured in the beginning of the proof. So $C_{\shell}'$  contains three or more stamens $s_1,\dots,s_k$ of $C$, and they are all adjacent to $v_i$.  Figure~\ref{fig:append_3} illustrates the construction of lines $\ell(s_1),\dots, \ell(s_k)$ in this last case. Notice that the angles between two consecutive half-lines $\ell(s_1),\dots,\ell(s_k)$ can be made less than $\pi$, since $v_i$ is not flat. Thus we can again proceed as in non-degenerate case. This concludes the proof of the claim.  
\end{description}

Let $\Gamma_{\shell}'$ be the constructed drawing of $G_{\shell}'$. Recall that each petal or stamen $w$ of $C$, that does not belong to $C_{\shell}'$, lies in a face of $G_{\shell}'$, denoted by $\shell(w)$. Let $\Gamma_{\shell(w)}$ denote the polygon representing face $\shell(w)$ in $\Gamma_{\shell}'$. By construction,  $\cone(w) \cap \Gamma_{\shell(w)} \neq \emptyset$. We next explain how to extend the drawing of $G_{\shell}'$ to the drawing of $G_{\shell}$.
For each edge $(u,v)$ of $C_{\shell}'$, we add a convex curve, lying close enough to this edge inside $\Gamma_{\shell}'$. Let $\mu$ be the union of these curves for all edges of $C_{\shell}'$. We notice that we can place them so close to $C_{\shell}'$ that all the points of  $\{\apex(w) \mid w\in C\}$ are still in the interior of $\mu$. Thus $\mu$ is intersected by all the sets $\cone(w)$, for each $w \in C$. We place  each vertex $w$ of $C_{\shell} \setminus C_{\shell}'$ on $\mu \cap \cone(w)$ in the order they appear in  $C_{\shell}$. Since all edges induced by $C_{\shell}$ lie outside of $C_{\shell}$, and both end points of such an edge are placed on a single convex curve, they can be drawn straight without intersecting each other, or other edges of $G_{\shell}$. Thus, we have constructed a planar extension of $\Gamma_C$ to a drawing of $G_{\shell}$, call it $\Gamma_{\shell}$.

Recall the definitions of faces of $C$, faces of $C_{\shell}$ and petal faces from the beginning of this section. 
The faces of $C$ appear in $\Gamma_{\shell}$ as convex polygons. The faces of $C_{\shell}$ are triangles, and the petal faces of $G_{\shell}$ are star-shapes whose kernel is close to the corresponding petal. 
By Observation~\ref{obs:G_structure}, each vertex of $G \setminus G_{\shell}$ either lies in a face of $C$, or in a face that is a triangle, or in a petal face, or outside $C_{\shell}'$. Moreover a subgraph of $G$ inside a petal face is triconnected. Thus, by multiple applications of Theorem~\ref{theorem:HongNagamochi}, we can extend the drawing of $G_{\shell}$ to a straight-line planar drawing of the subgraph of $G$ inside $C_{\shell}'$.  

Finally, notice that in the constructed drawing of $G_{\shell}$ each petal of its outer cycle, i.e. $C_{\shell}'$, is realizable. This is by the choice of edge $e$. Moreover, by construction of $G_{\shell}$, $C_{\shell}'$ has no outer chords. 
Thus, we can apply Theorem~\ref{theor:one-sided}, to complete the drawing of $G$, lying outside $C_{\shell}'$.
\qed
\end{proof}

We conclude with the following general statement, that follows from Theorem~\ref{theorem:convex} and one of the known algorithms that constructs drawing of a planar graph with a prescribed outer face (e.g. \cite{dgk-pdhgg-11,t-hdg-63} or Theorem~\ref{theorem:HongNagamochi}). 
\begin{corollary}
Let $G$ be a plane graph and $H$ be a biconnected plane subgraph of $G$. Let $\Gamma_H$ be a straight-line convex drawing of $\Gamma_H$. $\Gamma_H$ is extendable to a planar straight-line drawing of $G$ if and only if the outer cycle of $H$ is outerchordless and each petal of the outer cycle of $H$ is realizable. 
\end{corollary}

\section{Conclusions}

In this paper, we have studied the problem of extending a given convex
drawing of a cycle of a plane graph~$G$ to a planar straight-line
drawing of~$G$.  We characterized the cases when this is possible in
terms of two simple necessary conditions, which we proved to also be
sufficient.  We note that it is easy to test whether the necessary
conditions are satisfied in linear time.  It is readily seen that our
proof of existence of the extension is constructive and can be carried
out in linear time.  As an extension of our research it would be interesting to investigate whether more 
envolved necessary conditions are sufficient for more general shape of a cycle, for instance a star-shaped polygon.

\bibliographystyle{plain}
\bibliography{bibliography}
\end{document}